%% file: main-t-mss.tex
\documentclass{article}

\usepackage{amsmath}
\usepackage{amsthm}
\usepackage{newtxmath}
\usepackage{prettyref}
\usepackage[textsize=small,textwidth=4cm,color=green]{todonotes}
\usepackage{xcolor}
\usepackage{xspace}
\usepackage{fullpage}

\usepackage{thmtools}
\usepackage{thm-restate}

\usepackage{hyperref}
\usepackage{footnote}
\makesavenoteenv{tabular}

\DeclareMathOperator{\cl}{cl}
\newcommand{\bigdotcup}{\mathop{\dot\bigcup}}

\newcommand{\cost}{\text{cost}}
\newcommand{\cW}{\mathcal W}
\newcommand{\Ind}{\vvmathbb 1}
\let\1\Ind
\newcommand{\N}{\mathbb N}

\newcommand{\R}{\mathbb R}

\newcommand{\wfa}{\text{WFA}\xspace}
\newcommand{\cM}{\mathcal M}

\let\pref=\prettyref

\newtheorem{theorem}{Theorem}
\newtheorem{claim}[theorem]{Claim}

\newtheorem{lemma}[theorem]{Lemma}

\newtheorem{open}{Open problem}

\title{Metrical Service Systems with Transformations}

\author{
S{\'{e}}bastien Bubeck\thanks{Microsoft Research. Email: \texttt{sebubeck@microsoft.com}} \and
Niv Buchbinder\thanks{Tel Aviv University, Israel. Email: \texttt{niv.buchbinder@gmail.com}.}
\and
Christian Coester\thanks{CWI, Amsterdam, Netherlands. Email: \texttt{christian.coester@cwi.nl}.}
\and Mark Sellke\thanks{Stanford University. Email: \texttt{msellke@stanford.edu}}
}

\date{}

\begin{document}

\maketitle

\begin{abstract}
We consider a generalization of the fundamental online metrical service systems (MSS) problem where the feasible region can be transformed between requests. In this problem, which we call T-MSS, an algorithm maintains a point in a metric space and has to serve a sequence of requests. Each request is a map (transformation) $f_t\colon A_t\to B_t$ between subsets $A_t$ and $B_t$ of the metric space. To serve it, the algorithm has to go to a point $a_t\in A_t$, paying the distance from its previous position. Then, the transformation is applied, modifying the algorithm's state to $f_t(a_t)$. Such transformations can model, e.g., changes to the environment that are outside of an algorithm's control, and we therefore do not charge any additional cost to the algorithm when the transformation is applied. The transformations also allow to model requests occurring in the $k$-taxi problem.

We show that for $\alpha$-Lipschitz transformations, the competitive ratio is $\Theta(\alpha)^{n-2}$ on $n$-point metrics. Here, the upper bound is achieved by a deterministic algorithm and the lower bound holds even for randomized algorithms. For the $k$-taxi problem, we prove a competitive ratio of $\tilde O((n\log k)^2)$. For chasing convex bodies, we show that even with contracting transformations no competitive algorithm exists.

The problem T-MSS has a striking connection to the following deep mathematical question: Given a finite metric space $M$, what is the required cardinality of an extension $\hat M\supseteq M$ where each partial isometry on $M$ extends to an automorphism? We give partial answers for special cases.
\end{abstract}

\input{intro}

\input{pre}
\input{cd-lipschitz}

\section{Algorithms and Lower Bounds for Special Cases of T-MSS}\label{sec:special}

In this section we show upper and lower bounds on the competitive ratio for T-MSS for special cases.
In Section \ref{sec:convebody} we show that there exists no online algorithm with finite competitive ratio for nested convex body chasing with contractions in the plane, even with randomization.
In Section \ref{sec:taxi} we design a new randomized algorithm for the $k$-taxi problem. In Section \ref{sec:swap} we show upper and lower bounds for swap transformations. Finally, in Section \ref{sec:superl} we show a superlinear lower bound on the competitiveness of \wfa for isometry transformations.

\input{cd-lower}
\input{ktaxi}
\input{cd-swap}
\input{superlinear}

\input{blowup}

\bibliography{bibliography}{}
\bibliographystyle{plain}

\end{document}

%% file: intro.tex
\section{Introduction}
\emph{Metrical Service Systems (MSS)} \cite{ChrobakL91} is a fundamental online framework unifying countless problems. It also has a central role in our understanding of online computation and competitive analysis in general.
In this problem we are given a metric space $(M,d)$. The points of the metric represent possible states/configurations where an algorithm can serve requests; the distance between the states represents the cost of moving from one configuration to another. Each request consists of a subset of feasible states and the algorithm must serve the request by moving to one of these states. The cost of the algorithm for a sequence of requests is simply the total movement cost.

In the closely related problem of \emph{Metrical Task Systems (MTS)} \cite{BLS92}, each request is a cost function $c_t\colon M\to\R_+\cup\{\infty\}$. The algorithm can move to any point $x_t$, paying the movement cost as well as service cost $c_t(x_t)$. Note that MSS is equivalent to the special case of MTS where cost functions only take values $0$ and $\infty$, which already captures the essential difficulty of MTS. \footnote{Our results extend easily to the case where MTS requests are allowed, but we will stick to the MSS view for the sake of simplicity.} For deterministic algorithms, the competitive ratio on any $n$-point metric is $n-1$ for MSS~\cite{ManasseMS88} and $2n-1$ for MTS~\cite{BLS92}. For randomized algorithms, it lies between $O(\log^2 n)$ \cite{BCLL19,CL19} and $\Omega(\log n/\log\log n)$ \cite{BartalBM06,BartalLMN05}, and tight bounds of $\Theta(\log n)$ are known for some metrics. The MSS/MTS framework captures various central online problems such as paging, $k$-server, convex body chasing, layered graph traversal, etc. The competitive ratio of MSS usually serves as a first upper bound for the performance achievable for these special cases.

However, MSS fails to capture more dynamic environments in which configuration changes that are outside of the algorithm's control may occur. For example, new resources or constraints may appear/disappear and modify the configurations.
To capture such changes we propose an extension for MSS that allows transformations over the configuration space. For example, we may model the $k$-taxi problem\footnote{In this problem, introduced by \cite{FiatRR90}, there are $k$ taxis in a metric space. Each request is a pair of two points, representing the start and destination of a travel request by a passenger. Serving a request is done by selecting a taxi that travels first to its start and then its destination. In the hard version of the problem, the cost is defined as the total distance traveled by the taxis \emph{without carrying a passenger}.} by considering the possible configurations of taxis in the metric. A movement of a taxi from the start to the destination of a request simply corresponds to a transformation that maps any configuration that contains a taxi at the start to a configuration with an additional taxi at the destination and one less taxi at the start.
As these changes are dictated to any solution, it is reasonable not to account any cost for these changes both for the algorithm and for the offline (benchmark) solution.
In this work we initiate the study of {\em Metrical Service Systems with Transformations} (T-MSS), a generalization of the standard Metrical Service Systems (MSS) problem.
As before, we are given a metric space $(M,d)$. In each round $t$, we get a function (transformation) $f_t\colon A_t\to B_t$ that maps a subspace $A_t\subseteq M$ of feasible states to a subspace $B_t\subseteq M$. If $b_{t-1}\in M$ is the state of the algorithm before the request $f_t$ arrives, it has to choose one of the feasible states $a_t\in A_t$ and pays movement cost $d(b_{t-1},a_t)$. The new state of the algorithm is then $b_{t}:=f_t(a_t)$. The classical MSS problem is thus the special case of T-MSS where $f_t$ is the identity function on the set of feasible states at time $t$. In T-MSS we allow in addition to identity transformations also more complex transformations.
The high level question we ask is:
\begin{open}
What is the competitive ratio of the Metrical Service Systems with Transformations problem for families of metric spaces and allowable transformations?
\end{open}

\subsection{Our results and techniques}

We give partial answers to the above question, obtaining upper and lower bounds on the competitiveness for several interesting families of metric spaces and transformations.
The most general family of transformations we study are $\alpha$-Lipschitz transformations. 
Our main result is the following pair of (almost matching) upper and lower bounds for general metric spaces and $\alpha$-Lipschitz transformations.

\begin{restatable}{theorem}{cbounds}
\label{thm:cbounds}
There exists a deterministic $2\max\{2(\alpha+1),6\}^{n-2}$-competitive algorithm for T-MSS with $\alpha$-Lipschitz transformations on any $n$-point metric space.
Any algorithm for T-MSS with $\alpha$-Lipschitz transformations has competitive ratio at least $\min\{\alpha+1,\alpha^2\}^{n-2}$, even with randomization.
\end{restatable}

Although our results show an exponential lower bound (and an exponential upper bound) for any $\alpha>1$, they do not rule out linear/polynomial or with randomization even polylogarithmic competitive ratios when the transformations are $1$-Lipschitz (contractions). Resolving the competitive ratio for this important family of transformations is one of the most interesting remaining open questions.
On the other hand, we show that even restricting the transformations to be $1$-Lipschitz is not always enough.
In particular, when adding contraction transformations to the convex body chasing problem (that can be modeled as a special case of MSS on an infinite metric space) there exists no competitive algorithm even in the easier nested case. By contrast, with isometry transformations the competitive ratios $O(d)$ for the unrestricted problem and $O(\sqrt{d\log d})$ for the nested problem due to \cite{cbc6,cbc5,cbc4} remain unchanged as $\mathbb R^d$ is ultrahomogeneous (see following discussion for formal definitions).

\begin{restatable}{theorem}{cbc}
\label{thm:cbc}
There exists no online algorithm with finite competitive ratio for nested convex body chasing with contractions in the plane, even with randomization.
\end{restatable}

As a byproduct of our results we also get a new competitive algorithm for the $k$-taxi problem, which can be modeled as a T-MSS with isometry transformations.

\begin{restatable}{theorem}{ktaxi}
\label{thm:ktaxi}
There is a randomized $O((n\log k)^2\log n)$-competitive algorithm for the $k$-taxi problem on $n$-point metrics.
\end{restatable}

This result is better than the previous best bound of $O(2^k\log n)$ whenever $n$ is subexponential in $k$~\cite{CoesterK19}.

\paragraph{Extending partial isometries.}

As a basic tool to tackle T-MSS with general transformation, we study the problem when the allowable transformations are isometries. A map $f\colon A\to B$ for subsets $A,B\subseteq M$ is called a \emph{partial isometry of $M$} if it is distance-preserving, i.e., $d(f(x), f(y))= d(x, y)$. A metric space $M$ is called \emph{ultrahomogeneous} if every partial isometry of $M$ extends to an automorphism of $M$. Notice that on ultrahomogeneous metric spaces, the competitive ratio of T-MSS with isometry transformations is the same as the competitive ratio of MSS: Indeed, when a partial isometry $f_t\colon A_t\to B_t$ arrives, let $\hat f_t$ be its extension to an automorphism of $M$. In MSS, this request corresponds to the request $A_t$ followed by a renaming of the points of the metric space according to $\hat f_t$. Clearly, isometric renaming of points does not affect the competitive ratio.

If $M$ is not ultrahomogeneous, one could hope to extend $M$ to a larger metric space $\hat M$ such that every partial isometry of $M$ extends to an automorphism of $\hat M$. In this case, we call $\hat M$ a \emph{weakly ultrahomogeneous extension} of $M$. For a family of metric spaces $\cM$ we define the \emph{blow-up} as the supremum over all $n$-point metrics $M\in \cM$ -- of the minimum cardinality of a weakly ultrahomogeneous extension $\hat M$ of $M$.
If the blow-up can be bounded as a function of $n$ it allows us to apply the $(n-1)$-competitive deterministic algorithm or the $O(\log^2n)$-competitive randomized algorithm for MSS on the weakly ultrahomogeneous extension to get competitive algorithms. We study the blow-up also for restricted families of isometries. For example, we study swap isometries (defined on two points that are mapped to one another). We call a metric space \emph{swap-homogeneous} if every swap extends to an automorphism. Swap-homogeneous metric spaces have the intuitive property that the metric space ``looks'' the same at each point. Similar notions of metric space homogeneity have also been studied in other contexts (see, e.g., \cite{H52}). The main question, which we consider to be of independent interest, is thus:

\begin{open}
What is the blow-up for interesting families of metric spaces and partial isometries?
\end{open}

Some results on weakly ultrahomogeneous extensions already exist, in particular that every finite metric space has such an extension and it is of finite size \cite{solecki05,vershik08}. However, no bounds on the size of the extensions as a function of $n$ are known, and hence these results do not yield any bounds on the blow-up. We refer the reader to Section \ref{sec:blowup} for further discussion. Some of the algorithms we design for T-MSS are using these new upper bounds that we prove on the blow-up. The following theorem summarizes the upper and lower bounds we obtain.

\begin{restatable}{theorem}{bbounds}
\label{thm:bbounds}
The following bounds on the blow-up are tight.
\begin{center}
\begin{tabular}{l|l|l}
Family of metric spaces & Family of isometries & Blow-up \\
\hline
Ultrametrics & General& $2^{n-1}$  \\
Ultrametrics: $k$ distinct non-zero distances & General & $\approx\left(\frac{n+k-1}{k}\right)^{k}$\footnote{The precise blow-up is $(a+1)^b a^{k-b}$ where $a=\lfloor\frac{n+k-1}{k}\rfloor$ and $b=(n-1)\bmod k$.} \\
Equally spaced points on a line &  General & $2n-2$ \\
$\left(\{0,\dots,k\}^D,\text{weighted } \ell_1\right)$ & Translations & $(2k)^D$ \\
General metrics & Swaps & $2^{n-1}$
\end{tabular}
\end{center}
\end{restatable}

Using the direct reduction above, we may get directly some interesting results for T-MSS when transformations are isometries. For example, a randomized $O(n^2)$-competitive algorithm for ultrametrics and an $O(\log^2 n)$-competitive algorithm for equally spaced points on a line.

\paragraph{The work function algorithm.}

The \emph{work function algorithm} (\wfa) is a classical algorithm that achieves the optimal deterministic competitive ratio of $n-1$ for MSS on any $n$-point metric (see \cite{BorodinE98} for a discussion on its history). The algorithm extends naturally to T-MSS (see Section \ref{sec:workf}), and is a natural candidate algorithm to investigate. We prove that it is in fact optimal for several special cases of the problem. The following result for ultrahomogeneous ultrametrics is also used as part of our main algorithm for general metrics. We also prove that on general metrics, \wfa has a superlinear competitive ratio even for isometry transformations (which are $1$-Lipschitz). This may indicate that T-MSS has super linear competitive ratio even for isometry transformations.

\begin{theorem}\label{thm:workf}
The work function algorithm (\wfa) for T-MSS has competitive ratio
\begin{itemize}
\item $n-1$ on $n$-point ultrahomogeneous ultrametrics with $1$-Lipschitz transformations.
\item $2n-3$ on $n$-point metric spaces with swap transformations. No deterministic algorithm has competitive ratio better than $2n-3$ for this problem.
\item $\omega(n^{1.29})$ on some $n$-point metric space with isometry transformations, for each $n$.
\end{itemize}
\end{theorem}

\subsection{Organization}
In Section \ref{sec:pre} we formally define T-MSS, and discuss the work function algorithm (\wfa).
In Section \ref{sec:lip} we design an algorithm for T-MSS on general metrics with competitive ratio depending on the maximal Lipschitz constant of transformations, and prove an almost matching lower bound (Theorem \ref{thm:cbounds}). 
As part of the algorithm we also show that \wfa is $(n-1)$-competitive for ultrahomogeneous $n$-point ultrametrics and $1$-Lipschitz transformations, proving the first part of Theorem~\ref{thm:workf}.

In Section \ref{sec:special} we show upper and lower bounds on the competitive ratio for several special cases of T-MSS.
In Section \ref{sec:convebody} we show that there exists no online algorithm with finite competitive ratio for nested convex body chasing with contractions in the plane, even with randomization (Theorem \ref{thm:cbc}).
In Section \ref{sec:taxi} we design a new randomized algorithm for the $k$-taxi problem (Theorem \ref{thm:ktaxi}). In Section \ref{sec:swap} we show matching upper and lower bounds for swap transformations, proving the second part of Theorem \ref{thm:workf}. In Section \ref{sec:superl} we show a superlinear lower bound on the competitiveness of \wfa for isometry transformations, proving the third part of Theorem \ref{thm:workf}.
Finally, in Section \ref{sec:blowup} we prove upper an lower bounds on the blow-up for several families of metric spaces and transformations (Theorem \ref{thm:bbounds}). These results are also used earlier in the proofs of Theorem~\ref{thm:cbounds} and Theorem \ref{thm:ktaxi}.

%% file: pre.tex
\section{Preliminaries}\label{sec:pre}
In \emph{Metrical Service Systems with Transformations} (T-MSS), we are given a metric space $(M,d)$ and an initial state $b_0\in M$. We denote by $n$ the number of points in the metric space. In each round $t$, we get a function (transformation) $f_t\colon A_t\to B_t$ that maps a subset $A_t\subseteq M$ of feasible states to a subset $B_t\subseteq M$. If $b_{t-1}\in M$ is the state of the algorithm before $f_t$ arrives, it has to choose one of the feasible states $a_t\in A_t$ and pays movement cost $d(b_{t-1},a_t)$. The new state of the algorithm is then $b_{t}:=f_t(a_t)$. The classical MSS problem is thus a special case of T-MSS where $f_t$ is the identity function on the set of feasible states at time $t$. In T-MSS we always allow the identity transformations (thereby ensuring that T-MSS is a generalization of MSS) and in addition more complex transformations. We study the following important families of transformations $f:A \to B$ whose properties are defined below.

\begin{center}
\begin{tabular}{l|l}
Family of transformations & Condition\\
\hline
$\alpha$--Lipschitz & $\forall x,y\in A\colon d(f(x), f(y)) \leq \alpha \cdot d(x, y)$\\
$1$--Lipschitz (Contractions) & $\forall x,y\in A\colon d(f(x), f(y)) \leq d(x, y)$\\
Isometries & $\forall x,y\in A\colon d(f(x), f(y)) = d(x, y)$ \\
Swaps  & $A=\{a,b\}$, $f(a)=b, f(b)=a$\\
Translations & $M$ is subset of a vector space. $\forall x\in A\colon f(x) = x+v$ for a vector $v$
\end{tabular}
\end{center}

We also study several families of metric spaces. An important family of metric spaces are ultrametrics in which for every three points $x,y,z\in M$, $d(x,z) \leq \max\{d(x,y), d(y,z)\}$.
Ultrametric spaces may be viewed as the leaves of a rooted tree in which vertices with lowest common ancestor at level $i$ have distance $L_i$, where $0<L_1<L_2<\dots<L_k$ are the possible distances (where $k\leq n-1$).

\subsection{The work function algorithm for T-MSS}\label{sec:workf}
The \emph{work function algorithm} (\wfa) achieves the optimal deterministic competitive ratio of $n-1$ for MSS on any $n$-point metric. This algorithm is defined as follows: Denote by $p_0$ the fixed initial state. For some request sequence and a state $p\in M$, let $w_t(p)$ be the minimal cost to serve the first $t$ requests and then end up at $p$. The function $w_t$ is called the \emph{work function} at time $t$. We denote by $\cW$ the set of all maps $w\colon M\to\R_+$ that are $1$-Lipschitz. Notice that every work function is in $\cW$. The \wfa for MSS is the algorithm that, at time $t$, goes to a feasible state $p_{t}$ minimizing $w_t(p_t)+d(p_t,p_{t-1})$.

This algorithm extends naturally to T-MSS: Let $f_t\colon A_t\to B_t$ be the $t$th transformation. Let $w_t$ be defined as above and let $w_t^-(p)$ be the minimal cost of serving the first $t-1$ requests, then moving to some state in $A_t$ and then moving to $p$. Let $b_{t-1}$ be the state of the algorithm before time $t$. Upon the arrival of $f_t$, the \wfa first goes to a state
\begin{align*}
a_t\in\arg\min_{a\in A_t} w_t^-(a)+d(a,b_{t-1})
\end{align*}
and is then relocated to $b_t:=f_t(a_t)$.

We say that a work function $w\in \cW$ is \emph{supported} on a set $S\subseteq M$ if $w(x)=\min_{s\in S}w(s)+sx$ for each $x\in M$. The (unique) minimal such set $S$ is called the \emph{support} of $w$. Notice that $w_t^-$ is supported on $A_t$ and $w_t$ on $B_t$.

The following lemma is a variant of a lemma that is ubiquitous in analyses of \wfa for other problems~\cite{ChrobakL91}, adapted here to T-MSS:

\begin{lemma}\label{lem:extCost}
	Let $M$ be a metric space. Suppose there is a map $\Phi\colon\cW\to\R_+$ such that for any $x\in M$, time $t\ge 1$, $w\in\cW$, and any sequence of work functions $w_0,w_1^-,w_1,w_2^-,w_2,\dots$ arising for (a subclass of) T-MSS on $M$,
	\begin{align}
	w_t^-(x) - w_{t-1}(x)&\le \Phi(w_t^-)-\Phi(w_{t-1})\label{eq:PotBoundsExtCost}\\
	\Phi(w_t^-)&\le\Phi(w_t)\label{eq:PotIncr}\\
	\Phi(w)&\le \rho_M\cdot \min_{p\in M} w(p) + C_M,\label{eq:PotNumberWfs}
	\end{align}
	where $\rho_M$ and $C_M$ are constants depending only on $M$. Then \wfa is $(\rho_M-1)$-competitive for (this subclass of) T-MSS on $M$.
\end{lemma}
\begin{proof}
	We have
	\begin{align}
	w_t^-(b_{t-1})&=\min_{a\in A_t}w_t^-(a)+d(a,b_{t-1})\nonumber\\
	&= w_t^-(a_t)+d(a_t,b_{t-1})\nonumber\\
	&= w_t(b_t)+d(a_t,b_{t-1}),\label{eq:wfRec}
	\end{align}
	where the first equation is by definition of $w_t^-$ and the second equation by definition of $a_t$. The cost of \wfa is
	\begin{align*}
	\cost_\wfa&=\sum_{t=1}^T d(b_{t-1},a_t)\nonumber\\
	&= \sum_{t=1}^T \left(w_t^-(b_{t-1}) - w_t(b_t)\right)\nonumber\\
	&= \sum_{t=1}^T \left(w_t^-(b_{t-1}) - w_{t-1}(b_{t-1}) + w_{t-1}(b_{t-1}) - w_t(b_t)\right)\nonumber\\
	&\le \Phi(w_T) - w_T(b_T)\nonumber\\
	&\le (\rho_M-1)\cdot\min_{p\in M} w_T(p) + C_M,
	\end{align*}
	where the second equation follows from \pref{eq:wfRec}, the first inequality uses \pref{eq:PotBoundsExtCost}, \pref{eq:PotIncr}, $\Phi(w_0)\ge 0$ and $w_0(b_0)=0$, and the second inequality uses \pref{eq:PotNumberWfs}. Since $\min_{p\in M} w_T(p)$ is the optimal offline cost, the lemma follows.
\end{proof} 

%% file: cd-lipschitz.tex
\section{Competitivity for Lipschitz Transformations}\label{sec:lip}

In this Section we prove Theorem \ref{thm:cbounds}.

\cbounds*

The proof of the upper bound consists of three main steps: First we give a reduction to the case of $1$-Lipschitz transformations in ultrametrics. Then we employ the fact that ultrametrics admit an ultrahomogeneous extension of size $2^{n-1}$, which will be proved later in \pref{sec:ultrametrics}. Finally, we show that the \wfa achieves the optimal competitive ratio of $n-1$ for this special case of $1$-Lipschitz transformations on ultrahomogeneous ultrametrics.

\subsection{From $\alpha$-Lipschitz in general metrics to $1$-Lipschitz in ultrametrics}
For two metrics $d$ and $\hat d$ defined on a set $M$, we say that $\hat d$ is an \emph{$\alpha$-distortion} of $d$, for $\alpha\geq 1$, if for any $x,y\in M$ we have $d(x,y)\leq \hat d(x,y)\leq \alpha\cdot d(x,y)$.

\begin{lemma}\label{lem:distort}
Fix a constant $\alpha\geq 2$ and an $n$-point metric space $(M,d)$. There is an \emph{ultrametric} $\hat d$ on $M$ that is an $(\alpha+1)^{n-2}$-distortion of $d$, and any transformation $f\colon A\subseteq M\to M$ that is $\alpha$-Lipschitz with respect to $d$ is $1$-Lipschitz with respect to $\hat d$.
\end{lemma}

\begin{proof}

The idea is to group the edges (i.e., pairs of distinct elements) in $M$ into levels according to their distance. In level $1$ we start with the minimum distance edges, and repeatedly add all edges which are within a factor $\alpha$ of the largest level $1$ edge until no more are possible. We then continue, constructing level $k$ by starting with the shortest edge not in a previous level, and then adding all edges within a factor $\alpha$ of some edge already in level $k$. Let $L_k$ be the longest distance of an edge in level $k$. We define $\hat d$ by setting level $k$ edges to have $\hat d$-length $L_k$. 

Note that since $\alpha\geq 2$, being connected by edges of level at most $k$ is an equivalence relation for any $k$. Therefore, defining distances to be $L_k$ in level $k$ is a valid ultrametric. Since no edge has length in any interval $(L_k,\alpha L_k]$, any transformation that is $\alpha$-Lipschitz with respect to $d$ is $1$-Lipschitz with respect to $\hat d$. 

It remains to argue that distances were increased by at most a factor $(\alpha+1)^{n-2}$ in going from $d$ to $\hat d$. We show it for the edges in level $k$. Actually we may assume without loss of generality that $k=1$, because if not we may take all edges in levels up to $k-1$ and set their distances to the shortest edge distance in level $k$. This is still a valid metric and the case of level $1$ in this new metric implies the case of level $k$ in the old metric. For the main argument when $k=1$, consider adding the edges in order starting from an empty graph on the $n$ points of $M$, giving an increasing seqence $H_0,H_1,\dots$ of graphs. Exactly $n-1$ of these edges decrease the number of connected components by $1$ at the time they are added. We call these edges critical and denote by $D_1\le \dots\le D_{n-1}$ their lengths. At any time, the critical edges in $H_t$ form a spanning forest for $H_t$. Therefore the maximum distance of any edge in $H_t$ is at most the sum of the lengths of the critical edges in $H_t$. Therefore the next critical edge to be added has length

\[D_s\leq \alpha\cdot\sum_{r=1}^{s-1}D_r.\]

From this it is easy to see inductively that $D_s\leq \alpha(\alpha+1)^{s-2} D_1$ for $s\geq 2$. This means the maximum distance of any edge in $H$ is at most $\sum_{s}D_s\leq (\alpha+1)^{n-2}D_1$ (since there are $n-1$ critical edges). We remark that equality is achieved by the set of points $\{0,1,\alpha+1,(\alpha+1)^2,\dots,(\alpha+1)^{n-2}\}\subset\mathbb R$.
\end{proof}

\subsection{\texorpdfstring{\wfa for $1$-Lipschitz transformations on ultrahomogeneous ultrametrics}{\wfa for 1-Lipschitz transformations on ultrahomogeneous ultrametrics}}

We now prove that \wfa is $(n-1)$-competitive for $1$-Lipschitz transformations in ultrahomogeneous ultrametrics, thereby also proving the first statement of \pref{thm:workf}. Note that this bound is optimal, since $n-1$ is also the exact competitive ratio of ordinary MSS on any $n$-point metric space.

\begin{lemma}\label{lem:WfaUltra}
	Let $(M,d)$ be an ultrahomogeneous ultrametric with $n$ points. The \wfa for T-MSS on $(M,d)$ with $1$-Lipschitz transformation is $(n-1)$-competitive.
\end{lemma}
\begin{proof}
	We use the same potential function that also yields $(n-1)$-competitiveness for ordinary MSS on general $n$-point metrics and $(2n-1)$-competitiveness for MTS \cite{BorodinE98}:
	\begin{align*}
		\Phi(w):=\sum_{p\in M} w(p).
	\end{align*}
	The bound \eqref{eq:PotBoundsExtCost} follows from the fact that $w_t^-(p)- w_{t-1}(p)\ge 0$ for all $p$. Bound \eqref{eq:PotNumberWfs} for $\rho_M=n$ follows from the $1$-Lipschitzness of $w$, choosing $C_M$ to be $n-1$ times the diameter of $M$.
	
	For \eqref{eq:PotIncr}, we need to show that the sum of work function values is non-decreasing when a transformation $f_t\colon A_t\to B_t$ is applied. On a high level, the idea is as follows. The work function $w_t^-$ before transformation is supported on $A_t$, and the work function $w_t$ afterwards on $B_t$. Imagine for simplicity that the work function value of all support points were $0$. Then the work function value at other points is simply the distance from the support. Since $f_t$ is $1$-Lipschitz, we can think of $B_t$ as a contracted version of the set $A_t$ (and since the metric space is ultrahomogeneous, we can ignore the fact that $B_t$ might be located in a very different part of the metric space than $A_t$). This shrinking of the support means that most other points of the metric space tend to get further away from the support, and thus their work function values tend to increase.
	
	We now turn to a formal proof of inequality \eqref{eq:PotIncr}. Since $M$ is an ultrametric, we can view it as the set of leaves of an ultrametric tree.
	
	Denote by $T_r(p):=\{x\in M\colon d(p,x)\le r\}$ the ball of radius $r$ around $p$. Note that this is the set of leaves of the subtree rooted at the highest ancestor of $p$ whose weight is at most $r$. In particular, the sets $T_r(p)$ form a laminar family. We claim for all $a,a'\in A_t$ and $r,r'\ge 0$ that
	\begin{align}
		T_r(f_t(a))\cap T_{r'}(f_t(a'))	=\emptyset \implies T_r(a)\cap T_{r'}(a')=\emptyset.\label{eq:preserveContainment}
	\end{align}
	To see this, suppose $x\in T_r(a)\cap T_{r'}(a')$ and say without loss of generality that $r\le r'$. Then, by the ultrametric inequality, $d(a,a')\le \max\{d(a,x),d(a',x)\}\le r'$. Since $f_t$ is $1$-Lipschitz, this also means that $d(f_t(a),f_t(a'))\le r'$. But then $f_t(a)\in T_r(f_t(a))\cap T_{r'}(f_t(a'))$.
	
	For $y\ge 0$ and a work function $w$ supported on a set $S\in M$, its $y$-sublevel set is given by
	\begin{align*}
	w^{-1}([0,y])=\bigcup_{s\in S}T_{y-w(s)}(s).
	\end{align*}
	Recall that $w_{t}^-$ is supported on $A_t$, and note that $w_t$ is supported on a set $\tilde B_t\subseteq B_t$ such that
	\begin{align*}
	\forall b\in \tilde B_t \exists a_b\in A_t\colon f_t(a_b)=b\text{ and }w_t(b)=w_t^-(a_b).
	\end{align*}
	Since the sets $T_r(p)$ form a laminar family, we can choose for each $y\ge 0$ a subset $\tilde B_t^y\subseteq \tilde B_t$ such that
	\begin{align*}
	w_t^{-1}([0,y])=\bigdotcup_{b\in \tilde B_t^y}T_{y-w_t(b)}(b)
	\end{align*}
	is a disjoint union. Due to implication \eqref{eq:preserveContainment}, this also means that
	\begin{align*}
	\bigdotcup_{b\in \tilde B_t^y}T_{y-w_t^-(a_b)}(a_b)
	\end{align*}
	is a disjoint union.
	
	Now, the cardinality of the $y$-sublevel set of $w_t$ is bounded by
	\begin{align}
	|w_t^{-1}([0,y])| &= \sum_{b\in \tilde B_t^y}|T_{y-w_t(b)}(b)|\nonumber\\
	&= \sum_{b\in \tilde B_t^y}|T_{y-w_t^-(a_b)}(a_b)|\label{eq:cardSub2}\\
	&= \left|\bigdotcup_{b\in \tilde B_t^y}T_{y-w_t^-(a_b)}(a_b)\right|\label{eq:cardSub3}\\
	&\le \left|\bigcup_{a\in A_t}T_{y-w_t^-(a)}(a)\right|\nonumber\\
	&=|(w_t^-)^{-1}([0,y])|,\nonumber
	\end{align}
	where equation \eqref{eq:cardSub2} uses the fact that since $M$ is ultrahomogeneous, balls of the same radius have the same cardinality. Thus, for each $y\ge 0$ there are at least as many points whose $w_t^-$-value is at most $y$ as there are points whose $w_t$-value is at most $y$. Therefore, if $p_1,p_2,\dots,p_n$ and $p_1^-,p_2^-,\dots,p_n^-$ are two enumerations of $M$ by increasing $w_t$- and $w_t^-$-values, respectively, then $w_t^-(p_i^-))\le w_t(p_i)$. Hence, inequality \eqref{eq:PotIncr} follows.
\end{proof}

\paragraph{Remark.}
One may wonder whether the guarantee of \pref{lem:WfaUltra} is achieved more generally for ultrahomogeneous \emph{metric} (rather than ultrametric) spaces. The answer is negative: Consider the $8$-point metric space $M=\{0,2\}\times\{0,3\}\times\{0,4\}$ with the $\ell_1$-norm. Since $M$ is isometric to the cube $\{0,1\}^3$ with a weighted $\ell_1$-norm, and all partial isometries on $M$ are translations, \pref{thm:bbounds} (proved in \pref{sec:weightedL1}) implies that $M$ is ultrahomogeneous. Consider the work function $w$ supported $\{(0,0,0),(2,3,0)\}$, where it takes value $0$. We have $\Phi(w)=2+2+4+4+6+6=24$. The updated work function $w'$ after contracting $\{(0,0,0),(2,3,0)\}$ to $\{(0,0,0),(0,0,4)\}$ has $\Phi(w')=2\cdot(2+3+5)=20<\Phi(w)$, meaning that inequality \eqref{eq:PotIncr} is violated. Crucially, property \eqref{eq:preserveContainment} that disjoint balls remain disjoint when their centers are moved apart is violated on $M$. Using this observation, it is not hard to construct a request sequence on $M$ where \eqref{eq:PotBoundsExtCost} is always tight for $x=b_{t-1}$ and \eqref{eq:PotIncr} is either tight or violated for each request, and violated for a constant fraction of the requests. For such a request sequence, the analysis in the proof of \pref{lem:extCost} yields a \emph{lower} bound strictly larger than $n-1$ on the competitive ratio of \wfa for T-MSS on $M$ with $1$-Lipschitz transformations.

\subsection{Putting it together}
Given any $n$-point metric $(M,d)$, we obtain a $2\max\{2\alpha+2,6\}^{n-2}$-competitive deterministic algorithm for T-MSS with $\alpha$-Lipschitz transformations as follows: Making a multiplicative error of $\max\{\alpha+1,3\}^{n-2}$, Lemma \ref{lem:distort} allows us to assume that $M$ is an ultrametric and transformations are $1$-Lipschitz. By Theorem~\ref{thm:bbounds}, it further admits a $2^{n-1}$-point weakly ultrahomogeneous extension $(\hat M,\hat d)$, and the proof of this statement in Theorem~\ref{thm:bbounds} actually shows that $(\hat M,\hat d)$ is still an ultrametric and it is ultrahomogeneous (not just weakly). Therefore, by Lemma~\ref{lem:WfaUltra}, the \wfa is $2^{n-1}$-competitive on $(\hat M,\hat d)$. Overall, this gives a competitive ratio of $\max\{\alpha+1,3\}^{n-2}2^{n-1}=2\max\{2\alpha+2,6\}^{n-2}$.

\subsection{Lower bound for Lipschitz transformations}

In this section we prove the lower bound part of Theorem \ref{thm:cbounds} showing that any randomized algorithm for T-MSS with $\alpha$-Lipschitz transformations has competitive ratio at least $\min\{\alpha+1,\alpha^2\}^{n-2}$.

	Let $m:=\min\{\alpha+1,\alpha^2\}$. Assume $\alpha\ge 1$ since otherwise there is nothing to show. Consider the graph with vertices $p_1,\dots,p_n$ and edges from $p_1$ to every vertex and between consecutive vertices of lengths
	\begin{align*}
	d(p_1,p_i)&:=m^{i-2}&&i=2,\dots,n\\
	d(p_i,p_{i+1})&:=\alpha m^{i-2}&&i=2,\dots,n-1.
	\end{align*}
	Note that these edge lengths satisfy the triangle inequality: When adding the vertices to the graph in order, the addition of $p_i$ only creates the new triangle $(p_1,p_{i-1},p_i)$ with edge lengths $(m^{i-3},\alpha m^{i-3}, m^{i-2})$. Since $1\le \alpha\le m\le \alpha+1$, it satisfies the triangle inequality. Therefore, the shortest path extension of $d$ defines a valid metric.
	
	Consider a T-MSS instance on this space with $p_1$ as its initial state. For $t=1,2,\dots,n-2$, we issue transformations
	\begin{align*}
	f_{2t-1}\colon&\{p_1,p_{t+1}\}\to\{p_{t+1},p_{t+2}\}\\
	f_{2t}\colon&\{p_{t+1},p_{t+2}\}\to\{p_1,p_{t+2}\},
	\end{align*}
	where each transformation maps the first (resp. second) point of the domain to the first (resp. second) point of the codomain. Note that these maps are $\alpha$-Lipschitz, using for $f_{2t}$ that $m\le \alpha^2$. Then, we issue the final transformation
	\begin{align*}
	f_{2n-3}\colon \{x\}\to\{p_1\}
	\end{align*}
	with $x$ chosen uniformly at random from $\{p_1,p_n\}$.
	
	With probability $1/2$, the online algorithm has to pay $m^{n-2}$ to move to $x$ when $f_{2n-3}$ is issued. An offline algorithm can serve the sequence with expected cost $1/2$: Before $f_1$, it either stays at $p_1$ for cost $0$ (if $x=p_1$) or moves from $p_1$ to $p_2$ for cost $1$ (if $x=p_n$), which allows to serve the rest of the request sequence for free. Since the request sequence can be repeated arbitrarily often (with a new random $x$), we conclude a lower bound of $m^{n-2}$ on the competitive ratio.

%% file: cd-lower.tex
\subsection{Contracting convex bodies are unchaseable}\label{sec:convebody}

Here we show that nested convex body chasing with contractions has no competitive algorithm. In nested convex body chasing, the requests form a nested sequence $K_0\supseteq K_1\supseteq\dots $ of convex sets in $\mathbb R^d$. The player starts at $x_0\in K_0$ and moves online to a point $x_t\in K_t$, paying movement cost $\sum_{t\geq 1} ||x_{t-1}-x_t||.$ This problem is a special case of the more general convex body chasing problem which allows an arbitrary non-nested sequence $K_0,K_1,\dots$ of convex sets. This problem has received a lot of recent study without transformations and admits a $d$-competitive algorithm -- see \cite{cbc1,cbc2,cbc3,cbc4,cbc5,cbc6}. Because $\mathbb R^d$ is ultrahomogenous, it follows that the $d$-competitive algorithm continues to apply with partial isometry transformations.

Here we show that with contraction transformations $\mathbb R^d\to\mathbb R^d$ there is no competitive algorithm, even in the nested case for $d=2$ and with randomization. This gives a non-trivial example in which contractions are provably harder than isometries. In fact all the contractions we use are projections from $K_t$ to $K_{t+1}$. Our proof goes by reduction to a family of $1$-dimensional MTS problems known to have arbitrarily large competitive ratio. A related reduction appeared in \cite{kservercbc} to show that $2$-server convex body chasing is impossible in $2$ dimensions.

\cbc*

\begin{proof}

Fix a large integer $n$, a much larger integer $M=M(n)$ and a much larger $N=N(n)$. We start with $K_0=[0,1]\times [0,N]$ with starting point $x_0=(0,0)$. The projected convex sets rotate modulo $3$: At multiples of $3$ we simply have $K_{3t}=[0,1]\times [t,N].$ Now, for a sequence $(a_1,\dots,a_{N})$ of positive integers $a_i\in \{1,2,\dots,n\}$ define the points $p_1^t=(\frac{a_t}{n},t)$, $p_2^t=(\frac{a_t-1}{n},t+\frac{1}{Mn})$ and $p_3^t=(\frac{a_t+1}{n},t+\frac{1}{Mn}).$ We define $K_{3t+1}$ by cutting from $K_{3t}$ along the lines $p_1^tp_2^t$ and $p_1^tp_3^t$. We define $K_{3t+2}$ by also cutting along the line $p_2^tp_3^t$. See Figure~\ref{fig:CBC}.

We take $a_t$ to be an adversarially chosen sequence of such integers, yielding an adversarial nested chasing instance. For $s$ congruent to $0$ or $2$ modulo $3$, we apply projection maps (which are contractions) from $K_s$ onto $K_{s+1}$. Hence movement cost is incurred only on transitions from $K_{3t+1}\to K_{3t+2}$.

The idea is that the resulting problem is approximately a $1$-dimensional MTS, as without loss of generality we may assume $x_{3t}$ is always on the upper boundary of $K_{3t}$. It is easy to see that, crucially, the horizontal movement induced by the projections is $O(M^{-2})$ per time-step - we will treat this as an additive error term. The vertical movement corresponds to a metrical task system with cost functions given by $\frac{1}{M}c_{a_t}(x)$ for

\[c_{a_t}(x)=\begin{cases} 0 &\mbox{if } x \leq \frac{a_t-1}{n}\\
\min\{|x-\frac{a_t-1}{n}|,|x-\frac{a_t+1}{n}|\} &\mbox{if } x\in [\frac{a_t-1}{n},\frac{a_t+1}{n}]\\
0 & \mbox{if } x\geq \frac{a_t+1}{n}.\end{cases}\]

First, by repeating values of $a_t$ in blocks of $M$, we can obtain an MTS which directly uses cost functions $c_{a_t}(x)$, now with an additive error $O(M^{-1})$ per (block) time step for the horizontal movement. Next we claim any randomized algorithm for this new MTS can be assumed to stay on the finite set of values $\frac{k}{n}$ for $k\in \{0,1,\dots,n\}$. Indeed, if a randomized algorithm is at position $\frac{k+\alpha}{n}$ for $\alpha\in (0,1)$ we can instead move to $\frac{k}{n}$ with probability $1-\alpha$ and $\frac{k+1}{n}$ with probability $\alpha$. Moreover we can couple these roundings together by first sampling $u\in [0,1]$ uniformly and rounding at all times based on whether $u\leq \alpha$ or not. This turns any algorithm into a randomized algorithm which stays on the values $\frac{k}{n}$ and has the same expected movement cost. Because the cost functions $c_{a_t}$ are affine on each interval $(\frac{k}{n},\frac{k+1}{n})$ it results in the same expected service cost as well.

If the player is restricted to stay on the $n+1$ values $\frac{k}{n}$, the movement cost functions $c_{a_t}$ are $0$ at all but one of these values. Hence by repeating requests $O_n(1)$ times to force movement we may reinterpret this as an $n$-server problem on a metric space with $n+1$ points, by taking the player's location in the original problem to be the unique spot with no server. It is well-known \cite{BartalBM06,BartalLMN05} that the randomized competitive ratio of any $n$-server problem in a metric space with at least $n+1$ points is $\Omega(\log n/\log\log n)$. Finally, it is easy to see that for any MTS on a finite state space with competitive ratio $C$, an additive error in cost of $o(1)$ per time-step affects the competitive ratio by $o(1)$. Therefore for any fixed $n$, taking $M\to\infty$ results eventually in an MTS with competitive ratio $\Omega(\log n/\log\log n)$ even taking the horizontal effects of the projection maps into account. Finally taking $N$ sufficiently large to realize this competitive ratio gives the desired lower bound.
\end{proof}

\begin{figure}[h]
\centering
\includegraphics[width=8cm]{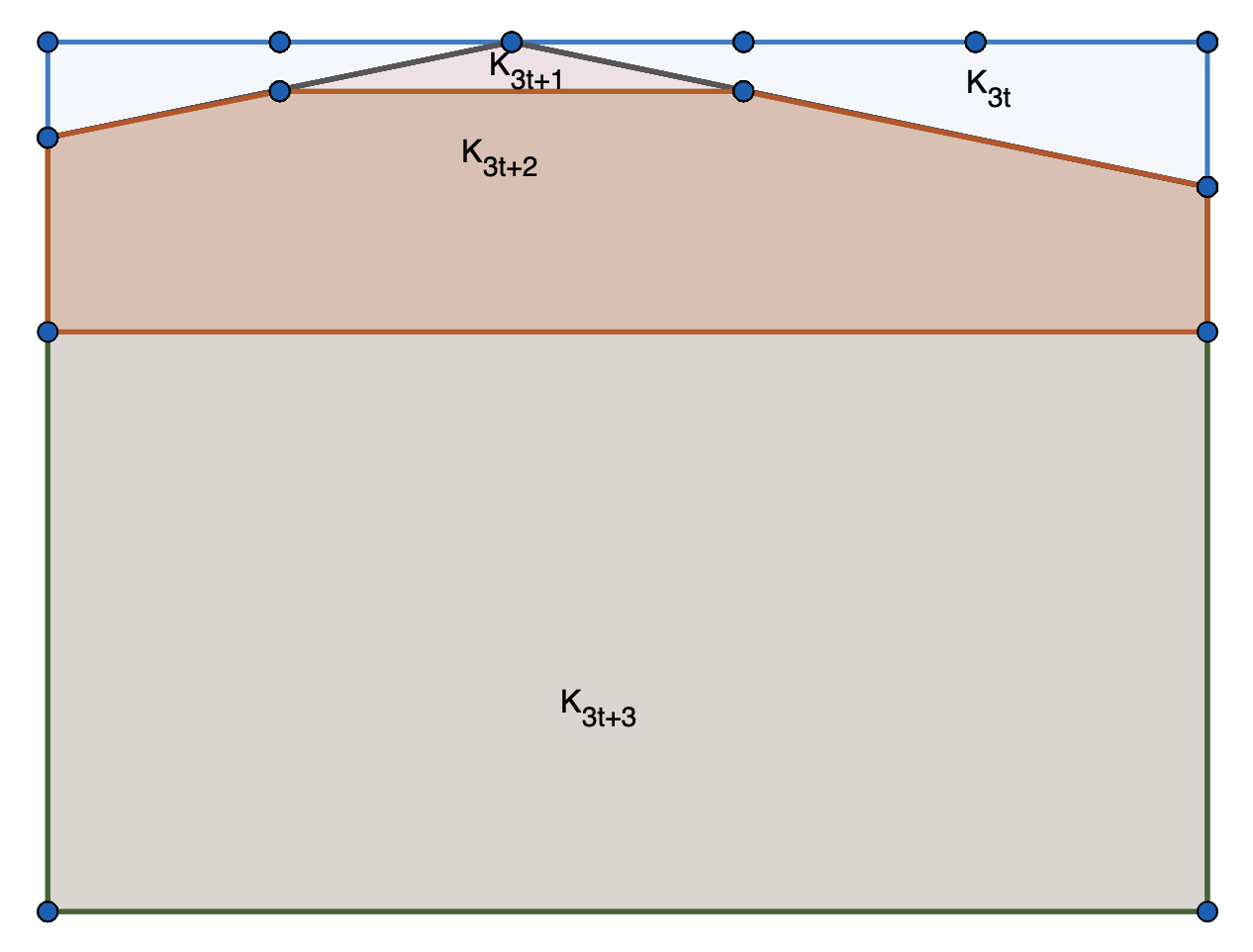}
\caption{To show that chasing nested convex bodies is impossible with contractions, we construct shrinking sets as shown. Euclidean nearest-point projections are taken onto the sets $K_{3t+1},K_{3t+3}$, so that movement cost is incurred only in moving from $K_{3t+1}$ to $K_{3t+2}$. Up to the negligible horizontal movements from projection, this results in a $1$-dimensional metrical task system with unbounded competitive ratio.}
\label{fig:CBC}
\end{figure} 

%% file: ktaxi.tex
\subsection{\texorpdfstring{A poly$(n,\log k)$-competitive $k$-taxi algorithm}{A poly(n,log k)-competitive algorithm for k-taxi}}\label{sec:taxi}

In the $k$-taxi problem, there are $k$ taxis located in a metric space $(M,d)$. A sequence of requests arrives, where each request is a pair of points $(s_t,d_t)\in M\times M$, representing the start and destination of a passenger request. Each request must be served upon its arrival by sending a taxi to $s_t$, from where it is relocated to $d_t$. The cost is defined as the distance travelled by taxis \emph{while not carrying a passenger}, i.e., excluding the distances from $s_t$ to $d_t$.

Note that the $k$-taxi problem is a special case of T-MSS: As the metric space for T-MSS, we take the set of taxi configurations (i.e., $k$-point multisets of points in $M$), and the distance between two configurations is the minimum cost of moving from one configuration to the other. A taxi request $(s_t,d_t)$ translates to the transformation that maps configurations containing $s_t$ to the corresponding configurations with a taxi at $s_t$ replaced by a taxi at $d_t$. Observe that this transformation is an isometry in the configuration space.

In principle, the size of a weakly ultrahomogeneous extension of the configuration space would yield a bound on the competitive ratio of the $k$-taxi problem. However, we do not know a bound on this blow-up in general. We overcome this obstacle as follows: First, we apply a well-known embedding of the original $k$-taxi metric space into a tree (HST) metric. Then, we consider the configuration space of this tree metric, whose metric is given by a weighted $\ell_1$-norm. Moreover, the isometries corresponding to the $k$-taxi requests in this tree metric are actually translations, and for this case, \pref{thm:bbounds} yields a bound on the blow-up (proved later in \pref{sec:weightedL1}). The resulting algorithm has competitive ratio $O((n\log k)^2\log n)$, improving upon the previous bound of $O(2^k\log n)$ \cite{CoesterK19} whenever $n$ is sub-exponential in $k$.

\ktaxi*
\begin{proof}
	By well-known techniques \cite{Bar96,Fakcharoenphol04}, any $n$-point metric space can be embedded with distortion $O(\log n)$ into the set of leaves of a random (weighted) tree. It therefore suffices to describe an $O((n\log k)^2)$-competitive algorithm for the $k$-taxi problem on the set $\mathcal L$ of leaves of a tree with $|\mathcal L|=n$. Notice that there is a tree $\mathcal T$ with only $O(n)$ vertices that induces the metric on $\mathcal L$.
	
	Let $V$ be the set of vertices of $\mathcal T$ excluding the root. For $v\in V$, let $w_v$ be the length of the edge from $v$ to its parent. If we denote by $x_v$ the number of taxis in the subtree rooted at $v$, then a configuration of $k$ taxis can be denoted by a point in $M:=\{0,\dots,k\}^V$. Notice that only some points in $M$ correspond to valid $k$-taxi configurations. The cost of moving from configuration $x$ to configuration $y$ is given by the metric
	\begin{align*}
		d(x,y):=\sum_{v\in V} w_v|x_v-y_v|.
	\end{align*}
	Thus, the $k$-taxi problem on $\mathcal T$ is a special case of T-MSS on $(M,d)$. A $k$-taxi request $(s_t,d_t)$ corresponds to a translation by the vector that has $1$-entries in coordinates of ancesters of $d_t$ that are not ancestors of $s_t$, $-1$-entries in coordinates of ancestors $s_t$ that are not ancestors of $d_t$, and $0$ in the remaining coordinates. We therefore need to extend $M$ to a space $\hat M$ where these translations extend to automorphisms. As stated in \pref{thm:bbounds} and proved later in \pref{sec:weightedL1}, such an extension $\hat M$ of size $(2k)^{|V|}$ exists. Thus, by running an $O(\log^2|\hat M|)$-competitive algorithm for $MSS$ on $\hat M$ and treating each automorphism as a renaming of the points of $\hat M$, we obtain an algorithm for the $k$-taxi problem on $\mathcal T$ with competitive ratio $O(\log^2|\hat M|)=O((n\log k)^2)$, where we used that $|V|=O(n)$. Combined with the $O(\log n)$ loss due to the tree embedding, the theorem follows.
\end{proof} 

%% file: cd-swap.tex
\subsection{Competitivity for swap transformations}\label{sec:swap}
In this section, we show tight bounds of $2n-3$ on the competitive ratio for T-MSS in general metrics when each transformation is either the identity on its domain (recall that we always allow identity transformations) or a swap, thereby proving the second part of \pref{thm:workf}. The upper bound is achieved by the \wfa.

\paragraph{Upper bound.} For brevity, we will denote the distance between two points $x,y\in M$ by $xy$ instead of $d(x,y)$. For $X\subseteq M$ and $x\in M$, we write $X-x:=X\setminus\{x\}$. For a work function $w$ and $x,y\in M$, let
	\begin{align*}
	\Psi_{x,y}(w)&:= \cl(M-x-y)+\sum_{p\in M-x-y}\min\{w(x)+yp,w(y)+xp\}\\
	\Phi(w)&:= \sum_{p\in M}w(p)+\min_{x,y\in M}\Psi_{x,y}(w)
	\end{align*}
	Here, $\cl(X):=\sum_{\{x,y\}\subseteq X} xy$ denotes the size of the clique of $X$, i.e., the sum of all distances between points in $X$.
	It suffices to show that $\Phi$ satisfies the properties of \pref{lem:extCost} with $\rho_M=2n-2$.
	
	Inequality \pref{eq:PotNumberWfs} is immediate from the fact that $\Phi(w)$ is a sum of $2n-2$ function values of $w$ and a bounded number of distances of $M$, and each function value of $w$ differs from $\min_x w(x)$ by at most the diameter of $M$ due to the $1$-Lipschitzness of $w$.
	
	Inequality \pref{eq:PotBoundsExtCost} follows from the fact that $\Phi(w)$ contains the summand $w(x)$, the remaining summands of $\Phi(w)$ are non-decreasing in $w$, and $w_{t-1}\le w_t^-$ pointwise.
	
	Inequality \pref{eq:PotIncr} is trivial if $f_t$ is the identity on its domain, so it remains to consider the case that $f_t\colon\{a,b\}\to\{a,b\}$ is a swap. Both $w_t^-$ and $w_t$ are supported on $\{a,b\}$. We first show that if $w$ is supported on $\{a,b\}$, then $\Psi_{x,y}(w)$ is minimized when $\{x,y\}=\{a,b\}$.
	
	Let $x$ and $y$ be such that $\Psi_{x,y}(w)$ is minimized and suppose $x\notin\{a,b\}$. Then we can assume without loss of generality (by symmetry) that $w(x)=w(a)+ax$. Then $\min\{w(x)+ya,w(y)+xa\}=w(y)+xa$. Thus,
	\begin{align*}
	\Psi_{x,y}(w)&=\cl(M-x-y)+w(y)+xa+\sum_{p\in M-x-y-a}\min\{w(a)+ax+yp,w(y)+xp\}\\
	&\ge\cl(M-x-y)+w(y)+xa\\
	&\qquad+\sum_{p\in M-x-y-a}\left(xp-ap+\min\{w(a)+yp,w(y)+ap\}\right)\\
	&\ge\cl(M-a-y)+\sum_{p\in M-y-a}\min\{w(a)+yp,w(y)+ap\}\\
	&= \Psi_{a,y}(w).
	\end{align*}
	Thus, $\Psi_{x,y}(w)$ is also minimized when $x=a$. If $y\ne b$ and $w(y)=w(b)+by$, then the symmetric argument shows that $\Psi_{x,y}(w)$ is minimized when $\{x,y\}=\{a,b\}$. Otherwise, if $y\ne b$, then $w(y)=w(a)+ay$ since $w$ is supported on $\{a,b\}$. Then
	\begin{align*}
	\Psi_{a,y}(w)&=\cl(M-a-y)+\sum_{p\in M-a-y}\min\{w(a)+yp,w(a)+ay+ap\}\\
	&= \cl(M-a)+(n-2)w(a)\\
	&\ge \cl(M-a-b)+\sum_{p\in M-a-b}\min\{w(a)+bp,w(b)+ap\}\\
	&= \Psi_{a,b}(w).
	\end{align*}
	Thus, it is indeed the case that $\min_{x,y}\Psi_{x,y}(w)=\Psi_{a,b}(w)$, for both $w=w_t^-$ and $w=w_t$. Hence,
	\begin{align*}
	\Phi(w_t^-)-\Phi(w_t)&=w_t^-(a)+w_t^-(b)-w_t(a)-w_t(b)\\
	&\qquad+\sum_{p\in M-a-b}\left(w_t^-(p)+\min\{w_t^-(a)+bp,w_t^-(b)+ap\}\right)\\
	&\qquad-\sum_{p\in M-a-b}\left(w_t(p)+\min\{w_t(a)+bp,w_t(b)+ap\}\right).
	\end{align*}
	Since $w_t(a)=w_t^-(b)$, $w_t(b)=w_t^-(a)$, and $w(p)=\min\{w(a)+ap,w(b)+bp\}$ for $w=w_t^-$ and $w=w_t$, everything cancels in the last sum and we obtain \pref{eq:PotIncr}.

\paragraph{Lower bound.} Consider the metric with points $1,2,\dots,n$ where the distance from $1$ to any other point is $1$ and the distance between any other two points is $2$. The initial location of the server is $1$. For $i=1,\dots,n-1$, the $i$th request is the identity with domain $A_i$, where $A_1:=\{2,3,\dots,n\}$ and for $i\ge 2$, $A_i$ is the subset of $A_{i-1}$ obtained by removing the location of the online algorithm's server before this request is issued. During these requests, the online algorithm suffers cost $2n-3$, but an offline algorithm could immediately go to the one point $p$ in $A_{n-1}$ for cost $1$. We issue one more request that swaps $p$ and $1$ so as to return to the initial configuration, allowing to repeat the procedure arbitrarily often.

%% file: superlinear.tex
\subsection{Superlinear lower bound for \wfa with isometries}\label{sec:superl}

In this section we show a superlinear lower bound of $\omega(n^{1.29})$ on the competitiveness of \wfa for T-MSS with isometry transformations, proving the third part of \pref{thm:workf}.

	It suffices to show the the statement for values of $n$ that are a power of $4$. Let $\alpha\in\N$ be some large constant. For $h\in\N_0$, we construct a $4^h$-point metric space $T_h$ by induction: The space $T_0$ is just a single point. For $h\ge 1$, space $T_h$ is a disjoint union of four copies of $T_{h-1}$, which we denote $T_{h-1}^0,T_{h-1}^1,T_{h-1}^2,T_{h-1}^3$. For points $x,y$ from two different copies of $T_{h-1}$ we define their distance as $\alpha^h$ if one of the copies is $T_{h-1}^0$ and as $2\alpha^h$ otherwise.
	
	For a set $X\subseteq T_{h-1}$ and $i=0,1,2,3$, denote by $X^i$ the copy of $X$ in $T_{h-1}^i$, and similarly if $X$ is a point rather than a set. We define a special point $s_h\in T_h$ as follows: $s_0$ is the single point in $T_0$. For $h\ge1$, $s_h:=s_{h-1}^0$.
	
	We consider T-MSS on $T_h$ when the server starts at $s_h$. Let $w_0=d(\,\cdot\,,s_h)$ be the initial work function. We will construct a request sequence $\sigma_h$ during which \wfa suffers cost $(6^h-1)\alpha^h(1-o(1))$ as $\alpha\to\infty$ and at whose end the work function is at most $\alpha^h+w_0$ pointwise, with \wfa returning to $s_h$ in the end. Since such a request sequence can be repeated, it will imply that the competitive ratio is at least $6^h-1=n^{(\ln 6) / (\ln 4)}-1=\omega(n^{1.29})$.
	
	For $h=0$, we simply choose the empty request sequence. Consider now $h\ge 1$. For a partial isometry $f\colon A\to B$ of $T_{h-1}$ and $i=0,1,2,3$, denote by $f^{i}\colon A^i\cup\bigcup_{j\ne i}T_{h-1}^j\to B^i\cup\bigcup_{j\ne i}T_{h-1}^j$ the map that acts like $f$ on $A^i$ and is the identity on $T_{h-1}^j$ for $j\ne i$. Note that $f^i$ is a partial isometry of $T_h$. Denote by $\sigma_{h-1}^i$ the sequence obtained by extending each partial isometry in $\sigma_{h-1}$ (the sequence from the induction hypothesis) in this way.
	
	We construct $\sigma_h$ as follows: First, we issue $2\alpha-2$ copies of $\sigma_{h-1}^{0}$. Since each work function $w$ during this sequence admits a point $p\in T_{h-1}^0$ with $w(p)\le (2\alpha-2)\alpha^{h-1}$, but $w(x)=\alpha^h$ for all $x\in\bigcup_{j\ne 0}T_{h-1}^j$, \wfa will stay within $T_{h-1}^0$ during these requests, suffering cost $(6^{h-1}-1)2\alpha^h(1-o(1))$. Then we issue the identity request with domain $\{s_{h-1}^1,s_{h-1}^2,s_{h-1}^3\}$, forcing the algorithm to move to one of these three points for cost $\alpha^h$. By symmetry, we can assume without loss of generality that \wfa moves to $s_{h-1}^1$. We now issue $2\alpha-2$ copies of $\sigma_{h-1}^{1}$, followed by the identity request with domain $\{s_{h-1}^2,s_{h-1}^3\}$. Similarly to before, \wfa again suffers cost $(6^{h-1}-1)2\alpha^h(1-o(1))$ and then moves to, say, $s_{h-1}^2$ for cost $2\alpha^h$. Finally, we issue $2\alpha-2$ copies of $\sigma_{h-1}^{2}$ followed by the request $\{s_{h-1}^3\}\to \{s_h\}, s_{h-1}^3\mapsto s_h$, increasing \wfa's cost by another $(6^{h-1}-1)2\alpha^h(1-o(1)) + 2\alpha^h$. Overall, \wfa suffers cost
	\begin{align*}
	(6^{h-1}-1)(2+2+2)\alpha^h(1-o(1))+(1+2+2)\alpha^h=(6^{h}-1)\alpha^h(1-o(1)),
	\end{align*}
	as claimed. Moreover, the final work function is $\alpha^h+w_0$ because an offline algorithm could move to $s_{h-1}^3$ for cost $\alpha^h$ at the start of the request sequence, suffer no cost during the rest of the sequence, and be mapped back to $s_h$ (for free) via the final request.

%% file: blowup.tex
\section{Bounds on the Metric Extension Blow-up}\label{sec:blowup}

In this section we prove upper an lower bounds on the blow-up for several families of metric spaces and transformations,
proving Theorem \ref{thm:bbounds}.

\bbounds*

It was shown independently by Solecki~\cite{solecki05} and Vershik~\cite{vershik08} that every \emph{finite} metric space $M$ admits a \emph{finite}\footnote{The existence of an \emph{infinite} ultrahomogeneous metric space that extends every finite metric space, called the Urysohn universal space, has been known since the 1920s.} weakly ultrahomogeneous extension $\hat M$. An elementary proof of this result was presented very recently in \cite{hubickaKN18}. However, the main part of the construction in \cite{hubickaKN18} consists of $\lceil R\rceil$ growing steps, where $R$ is the aspect ratio of $M$, and a naive bound on the growth factor in the $i$th step alone is already doubly exponential in $i$. Thus, this does not yield an upper bound on the cardinality of $\hat M$ in terms of the cardinality of $M$, but only one that also involves the aspect ratio. Thus, even though there exists a finite extension for any $n$-point metric, it is unclear whether its size can be bounded as a function of $n$. If not, this would mean that the blow-up for general metrics and isometries is infinite.

\subsection{General metrics with swap transformations}

Let $M$ be an $n$-point metric. We will show how to extend $M$ to a $2^{n-1}$-point space where every swap extends to an automorphism. The tightness of this upper bound follows from the lower bound for ultrametrics (with $n-1$ distinct distances) proved in Section~\ref{sec:ultrametrics}. There, we will show that even if only maps with $1$-point domain need to extend to automorphisms, the extended space may require cardinality $2^{n-1}$.

We embed $M$ into the vector space $\hat M=\mathbb F_2^{n-1}$ by enumerating the points $p_1,\dots,p_n$ of $M$ in arbitrary order and defining the embedding $\varphi:p_k\mapsto (1^{k-1},0^{n-k})$. We choose the metric on $\hat M$ to extend that of $M$ and also be translation invariant, and explain just below why such a choice exists. Now for $x,y\in \mathbb F_2^{n-1}$, their swap is translation by $(x+y)$. Since the metric on $\hat M$ is translation invariant, this translation gives the desired extension to an automorphism on $\hat M$.

Now we explain why there exists such a translation invariant metric on $\hat M$. We first extend to a partial metric $(\hat M,\hat d)$ only by translation invariance, i.e., $\hat d$ is the partial function on $\hat M\times \hat M$ defined by $\hat d(x,y)=d(p_i,p_j)$ if $x-y=\varphi(p_i)-\varphi(p_j)$. Viewing $\hat d$ as a weighted graph $G$, the shortest path extension of $\hat d$ is also translation invariant. Moreover, it gives a valid metric on $\hat M$ if and only if there is no cycle in $G$ violating the triangle inequality.

Now, the values $\varphi(p_i)-\varphi(p_j)\in\hat M$ range over vectors with a single continuous block of $1$s. A cycle in $G$ consists of a multiset $\mathcal S$ of such vectors adding to $0$. Viewing $\varphi(p_i)-\varphi(p_j)$ as an edge connecting $p_i$ and $p_j$, we claim that any such multiset $\mathcal S$ must correspond to a multigraph on $M$ with even degree at each vertex $p_i$. Indeed, for $i\ge 2$, the parity of the degree at $p_i$ is the difference between the coordinates $i-1$ and $i$ in the summed vector, which is $0$ by definition (treat coordinate $n$ as being identically $0$). Since the sum of degrees is even, also $p_1$ must have even degree. Therefore this multigraph has an Eulerian circuit and in particular a cycle containing each edge. Since the edge weights are now exactly distances in $M$, we are done by applying the triangle inequality in $M$.

\subsection{Ultrametrics}\label{sec:ultrametrics}
We consider ultrametrics with at most $k$ distinct non-zero distances. Note that any $n$-point ultrametric has at most $n-1$ distinct distances, so the general bounds on ultrametrics follow from the case $k=n-1$.

\paragraph{Upper bound.}
Recall that ultrametric spaces may be viewed as the leaves of a rooted tree in which vertices with lowest common ancestor at level $i$ have distance $L_i$, where $0<L_1<L_2<\dots<L_k$ are the possible distances.

We construct $\hat M$ by augmenting the $n$-leaf tree corresponding to $M$ with additional points to create a symmetric tree $\hat M$ where partial isometries extend to automorphisms. We claim this can be done with at most $\left(\frac{n+k-1}{k}\right)^{k}$ total leaves. Indeed, the original tree's non-leaf vertices have $C_1,\dots,C_j$ children for some numbers satisfying $\sum_{i=1}^j (C_j-1)=n-1$. Therefore letting $E_i$ be the maximal number of children for any vertex at level $i$, we have $\sum_i (E_i-1)\leq n-1$. We take $\hat M$ to be the leaves of a fully symmetric tree in which every vertex at level $i$ has $E_i$ children. It is clear that any partial isometry of this symmetric tree extends to an automorphism. Moreover $|\hat M|=\prod_i E_i$. Given the constraint $\sum_{i=1}^k (E_i-1)\leq n-1$, the bound $|\hat M|\leq \left(\frac{n+k-1}{k}\right)^{k}$ follows from AM-GM. Since each $E_i$ is an integer, the precise bound is $(a+1)^b a^{k-b}$ if $n-1=ak+b$ for $b\in \{0,1,\dots,k-1\}$.

\paragraph{Lower bound.}
Let $M_0,\dots,M_k$ be disjoint sets, where $M_0=\{p_0\}$ is a singleton and $M_1,\dots,M_k$ have cardinalities $\lfloor\frac{n-1}{k}\rfloor$ or $\lceil\frac{n-1}{k}\rceil$ such that the union $M:=M_0\dot\cup M_1\dot\cup\dots\dot\cup M_k$ has cardinality $n$. We define an ultrametric on $M$ by defining the distance between any two distinct points $x\in M_i$, $y\in M_j$ with $i\le j$ to be $2\cdot 3^{j-1}$.

Let $\hat M\supseteq M$ be a weakly ultrahomogeneous extension of $M$. For a point $x\in \hat M$ and $j=0,\dots,k$, denote by $B_j(x)$ the set of points in $\hat M$ within distance strictly less than $3^j$ from $x$. We claim that $|B_j(p_0)|\ge \prod_{i=1}^{j}(|M_i|+1)$ for each $j=0,\dots,k$. This implies that $|\hat M|\ge \prod_{i=1}^k(|M_i|+1)\ge \prod_{i=1}^k\lfloor\frac{n+k-1}{k}\rfloor$.

To prove the claim, we proceed by induction on $j$. Clearly it is true for $j=0$. For $j\ge 1$, consider the balls $B_{j-1}(p)$ for $p\in M_j\cup\{p_0\}$. By the triangle inequality, they are disjoint and contained in $B_j(p_0)$. Since for each $p\in M_j$ the map $p_0\mapsto p$ extends to an isomorphism, they must also all have the same cardinality, which is at least $\prod_{i=1}^{j-1}(|M_i|+1)$ by the induction hypothesis. Thus, $B_j(p_0)$ has cardinality at least $\prod_{i=1}^{j}(|M_i|+1)$.

\subsection{The line and weighted $\ell_1$-norms with translations}\label{sec:weightedL1}

The extension of $n$ equally spaced points on a line is simple: We extend it to a circle of $2n-2$ equally spaced points. It is easy to see that the circle is ultrahomogeneous because any (partial) isometry is a combination of a rotation and possibly a reflexion. We will now extend this idea to multiple dimensions.

Consider the space $M:=\{0,1,\dots,k\}^D$ with the distance given by the weighted $\ell_1$-norm $d(x,y):=\sum_{i=1}^D w_i|x_i-y_i|$, where $w_1,\dots,w_D$ are arbitrary positive weights. As partial isometries, we consider the family of translations $x\to x+v$ that map a subset of $M$ to another subset of $M$. We will show that the associated blow-up is precisely $(2k)^D$. Note that the lower bound for $D=1$ also yields a tight lower bound of $2n-2$ for the blow-up of equally spaced points on a line,

\paragraph{Upper bound.} Notice that any translation is a composition of translations of the form $x\to x+e_i$ and their inverses, where $e_i\in\{0,1\}^D$ is the vector with a $1$-entry in only the $i$th coordinate. It therefore suffices to extend $M$ to a metric space $\hat M$ where partial isometries of this restricted type extend to global isometries.

We extend $M$ to the space $\hat M=\{0,\dots,2k-1\}^D$ and define a metric on $\hat M$ by
\begin{align*}
\hat d(x,y):=\sum_{i=1}^D w_i\min\{|x_i-y_i|,2k-|x_i-y_i|\}.
\end{align*}
This is the metric induced by the weighted $\ell_1$-norm when viewing $\hat M$ as a $D$-dimensional torus. Clearly, $\hat d$ extends $d$. Moreover, any isometry $x\to x+e_i$ defined on a subset of $M$ extends to the automorphism $x\to x+e_i \text{ mod }2k$ on $\hat M$, where the ``$\text{mod }2k$'' is applied coordinate-wise.

\paragraph{Lower bound.} Let $A_0:=\{0,1,\dots,k\}$ and $A_1:=\{0,1\}$. For a $0$-$1$-string $i_1i_2\dots i_D$, consider the translation
\begin{align*}
f_{i_1\dots i_D}\colon &A_{i_1}\times \dots\times A_{i_D} \to M\\
&x\mapsto x+(k-1)\cdot(i_1,\dots, i_D).
\end{align*}
The choice of domain of $f_{i_1\dots i_D}$ is just to ensure that the image is still in $M$.

Let $\hat M\supseteq M$ be an extension of $M$ such that each $f_{i_1\dots i_D}$ extends to an automorphism $\hat f_{i_1\dots i_D}$ of $\hat M$.

Let $C_0:=\{0,1,\dots,k-1\}$, $C_1:=\{1,2,\dots,k\}$ and $S_{i_1\dots i_D}:=\hat f_{i_1\dots i_D}(C_{i_1}\times\dots\times C_{i_D})$.

Note that each set $S_{i_1\dots i_D}$ has cardinality $k^D$, and there are $2^D$ such sets in total, corresponding to the $2^D$ possible $0$-$1$-strings of length $D$. Thus, the lower bound of $(2k)^D$ on the cardinality of $\hat M$ follows from the following claim.

\begin{claim}
	The sets $S_{i_1\dots i_D}$ are pairwise disjoint for different $0$-$1$-strings $i_1\dots i_D$.
\end{claim}
\begin{proof}
	Let $y\in S_{i_1\dots i_D}$ for some $0$-$1$-string $i_1\dots i_D$. We will show that $i_1\dots i_D$ is uniquely determined by $y$.
	
	We can write $y=\hat f_{i_1\dots i_D}(x)$ for some $x\in M$. It suffices to show that $i_j=0$ if and only if $y$ is closer to $k\cdot\1-e_j$ than to $k\cdot\1$, where $\1$ denote the all-ones vector. Equivalently, we will show that
	\begin{align*}
	i_j=0\iff d(x, f_{i_1\dots i_D}^{-1}(k\cdot\1-e_j))< d(x, f_{i_1\dots i_D}^{-1}(k\cdot\1)).
	\end{align*}
	Note that the preimages $f_{i_1\dots i_D}^{-1}(k\cdot\1)$ and $f_{i_1\dots i_D}^{-1}(k\cdot\1-e_j)$ exist in $M$, and they differ only in their $j$th entry.
	
	If $i_j=0$, then $x_j\le k-1$ (by definition of $C_{i_j}$) and the $j$th entries of $f_{i_1\dots i_D}^{-1}(k\cdot\1-e_j)$ and $f_{i_1\dots i_D}^{-1}(k\cdot\1)$ are $k-1$ and $k$, respectively. Thus, $x$ is closer to $f_{i_1\dots i_D}^{-1}(k\cdot\1-e_j)$ than to $f_{i_1\dots i_D}^{-1}(k\cdot\1)$.
	
	If $i_j=1$, then $x_j\ge 1$ (by definition of $C_{i_j}$) and the $j$th entries of $f_{i_1\dots i_D}^{-1}(k\cdot\1-e_j)$ and $f_{i_1\dots i_D}^{-1}(k\cdot\1)$ are $0$ and $1$, respectively. Thus, $x$ is further from $f_{i_1\dots i_D}^{-1}(k\cdot\1-e_j)$ than from $f_{i_1\dots i_D}^{-1}(k\cdot\1)$.
\end{proof}